\documentclass[journal]{IEEEtran} 

\usepackage[pdftex]{graphicx}
\usepackage{algorithm}
\usepackage[noend]{algpseudocode}
\usepackage{booktabs}
\usepackage{amssymb}
\usepackage{amsthm,xpatch}
\usepackage{blkarray}
\usepackage{multirow} 
\usepackage[cmex10]{amsmath}
\usepackage{subcaption}
\usepackage{lipsum}
\usepackage{hyperref}
\usepackage{xcolor}
\usepackage{bbm}
\usepackage{gensymb}
\usepackage[noadjust]{cite}

\newtheorem{theorem}{Theorem}

\makeatletter
\xpatchcmd{\proof}{\topsep6\p@\@plus6\p@\relax}{}{}{}
\makeatother

\begin{document}
	
	\title{Deep Learning based Coverage and Rate Manifold Estimation in Cellular Networks}

	\author{Washim Uddin Mondal, Praful D. Mankar, Goutam Das, Vaneet Aggarwal, and Satish V. Ukkusuri\thanks{W. U. Mondal is with the Schools of CE and IE, Purdue University, West Lafayette, IN 47907, USA, email: wmondal@purdue.edu. P. D. Mankar is with the Signal Processing and Communication Research Center, IIIT, Hyderabad, India 500032, email: praful.mankar@iiit.ac.in. G. Das is with the G. S. Sanyal School of Telecommunications, IIT Kharagpur, India 721302, email: gdas@gssst.iitkgp.ac.in. V. Aggarwal is with the Schools of IE and ECE, Purdue University, West Lafayette, IN 47907, USA, email:  vaneet@purdue.edu. S. V. Ukkusuri is with the School of CE,  Purdue University, West Lafayette, IN 47907, USA, email: sukkusur@purdue.edu.}
	}
	
	\maketitle
	
	\begin{abstract}
		This article proposes Convolutional Neural Network based Auto Encoder (CNN-AE) to predict location dependent rate and coverage probability of a network from its topology. We train the CNN utilising BS location data of India, Brazil, Germany and the USA and compare its  performance  with stochastic geometry (SG) based analytical models. In comparison to the best-fitted SG-based model, CNN-AE improves the coverage and rate prediction errors by a margin of as large as $40\%$ and $25\%$  respectively. As an application, we propose a low complexity, provably convergent algorithm that, using trained CNN-AE, can compute locations of new BSs that need to be deployed in a network in order to satisfy pre-defined spatially heterogeneous performance goals.
	\end{abstract}

	
	\IEEEpeerreviewmaketitle
	%\fontsize{12.5pt}{15pt}\selectfont
	
	\section{Introduction} 
	The relationship between the topology of a network and its  performance is one of the most important questions in cellular network industry \cite{lu2021stochastic}. An answer to this question will not only allow us to evaluate the performance of existing networks but also provide a way to design smart networks for the future with enhanced performance. With rapid development of augmented and virtual reality, video streaming, network gaming and other data-hungry applications, the issue of smart design has become pertinent for the forthcoming fifth generation (5G) and beyond networks \cite{wijethilaka2021survey}. There are two major approaches in the literature, namely simulation and stochastic geometry, that are applied to map a cellular network topology to its performance. \hspace{1cm}
	
	\subsection{Simulation}
	In a simulation based setup, the precise locations of the base stations (BSs) are fed into a \textit{simulator} along with a statistical model of the channel gains. The simulator then calculates the Signal-to-Interference-plus-noise ratio (SINR) at each possible user location for a given realisation of the BS-to-user channels. 
	The task is then repeated sufficiently large number of times to generate a \textit{user location-dependent} probabilistic description of a desired performance metric \cite{ozcan2020robust}.
	For example, one can obtain location-specific coverage probability (i.e., the probability that the SINR at a given location surpasses a predefined threshold) and  location-specific average rates via this method.
	
	Simulation provides an accurate description of network reality. However, the simulation process is quite computationally expensive and therefore, time-consuming, especially where the number of BSs and potential user locations are high \cite{sabbah2018emulation}. Thus, simulation cannot be efficiently used for network design where we must simulate large number of network scenarios to choose the best performing network.
	
	\subsection{Stochastic Geometry}
	
	Stochastic Geometry (SG) is an alternate tool to simulation that analytically evaluates the network performance. One of the core presumptions of SG-based models is that the locations of the users and BSs in a network can be described as realizations of \textit{stationary two-dimensional (2D) random point processes}. A 2D point process is essentially a collection of random points in an infinite 2D area such that their locations are determined by a given probability law. For example, 2D Poison Point Process (PPP) was one of the earliest models for BS and user locations \cite{andrews2011tractable}. It is defined as a collection of random points in a 2D plane such that the number of points contained in two disjoint areas can be depicted as two independent Poisson random variables. On the other hand, a 2D random process is said to be stationary if translation of each random point by a fixed vector does not change the probabilistic description of the process \cite{chen2018modeling}. PPP is an example of a stationary point process.
	
	An advantage of a stationary process is that it appears to be \textit{statistically} identical from the perspective of \textit{every} arbitrarily chosen points in the plane. Hence, if BS locations are depicted by a stationary process, then in the ensemble of all possible BS location realisations, the average network performance must be identical at each point in the plane. The appeal of the SG-based models is that they yield closed-form expressions for location-independent ensemble-averaged network performance. \hspace{1cm}
	
	However, there are two major pitfalls to this approach. First, in a network design scenario, one is typically interested in the performance of a given network realization rather than the that of their ensemble. In a given network realization, different user may experience different network performance. Moreover, due to various socio-economic, and demographic reasons, different user locations may have different network design requirement. For example, the highly populated areas of a city may require higher coverage probability whereas the low-populated places may need lower coverage \cite{li2018stochastic}. It is clear that to design networks with such spatially heterogeneous goals, we must have access to location-specific network performance and hence, location-independent SG-based average performances are not useful for this purpose.

	Secondly, the assumed BS location models in the SG-based analyses may not always be realistic. For example, PPP allows the BSs to be placed arbitrarily close to each other. Repulsion based models \cite{deng2014ginibre} were later introduced to avoid this drawback. Clustered BS process-based analyses \cite{afshang2018poisson} also appeared in the literature to mimic the clustered real networks. Comparing the fitted error of a pool of candidate point processes, it was later established \cite{li2018stochastic} that the $\alpha-$Stable process best depicts  the BS deployment in urban scenarios. All of these refinements added realism to the model, however, at the expense of increasingly complex average performance analysis. Also, with increase in complexity, the estimation process of different parameters of the models from the real data became challenging.
	
	\begin{figure*}[h!]
		\centering
		\includegraphics[width=.9\linewidth]{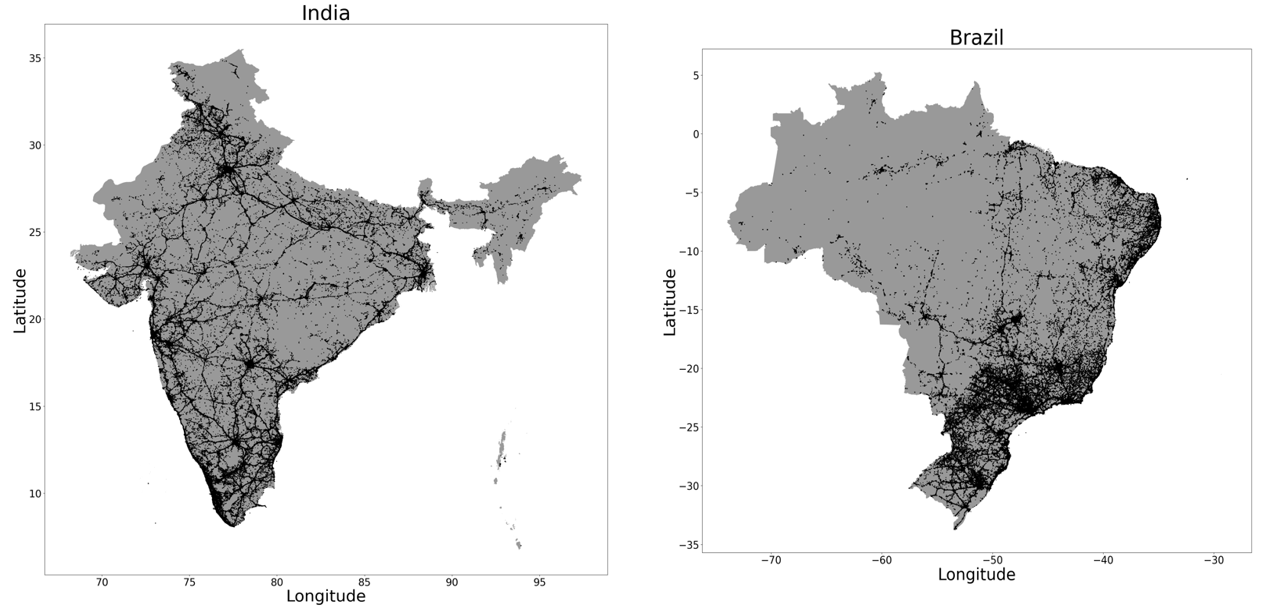}
		\caption{BS locations in India and Brazil. BS locations of Germany and the USA are not shown as those are too dense to make a meaningful visual depiction.}
		\label{fig:joint_vis}
	\end{figure*}
	
	\subsection{Our Approach}
	The gist of the above discussion is that simulation caters to network design needs but demands large computation-time. On the other hand, SG-based models are computationally fast but cannot provide location-specific network performance. In this article,
	we take the middle path and come up with a procedure that can provide location specific network performance despite being easily executable. In other words, our target is to mimic the function of a simulator in a computationally-efficient way.
	
	At a high level, the task of a simulator can be described as a functional mapping that takes a network realization as an input and produces its location-dependent performances (which we collectively define as the \textit{performance manifold}) as an output. In this article, we use a Neural Network (NN) to approximate this mapping. Sufficiently dense NNs with proper architecture and adequate training can approximate any functional mapping with reasonable accuracy. Therefore, it is currently being used as a universal translator between different data types in various areas-from language translation \cite{otter2020survey} to image recognition \cite{xie2020adversarial}. In  wireless communication, NN based Deep Learning models are used for channel prediction \cite{yuan2020machine}, trajectory planning \cite{zhu2021uav} and resource allocation \cite{xu2021experience} among others.

	To the best of our knowledge, this work is the first to explore the concept of NN-based translation of a network topology to its performance manifold. We feed  the BS locations of India, Brazil, Germany, and the USA to our NN as binary images and train it to estimate coverage and rate manifolds of any network located in the countries mentioned above.
	We use the outputs of the simulator as the ground truth for the supervised training. For benchmarking, we compare the performance of our model to that of the fitted SG-based analytical models. Interestingly, we notice that even a simple NN-architecture can improve the coverage and rate manifold prediction errors by a margin of as large as $40\%$ and $25\%$ respectively in comparison to the SG- based models. As an application, we also develop an algorithm that, using the trained NN, determines the locations where new BSs must be deployed in a region so as to achieve a predefined spatially heterogeneous performance goal.
	
	In summary, our key contributions can be listed as follows.
	\begin{itemize}
		\item We construct a Convolutional Neural Network based Auto Encoder (CNN-AE) to predict coverage and rate manifold of a cellular network from its BS locations.
		\item We compare the performance of CNN-AE to SG baseline models and note a dramatic reduction in prediction error.
		\item As an application, in section \ref{network_design}, we develop an algorithm that, using the trained CNN-AE, can efficiently determine the locations of deployable BSs in a brownfield network design problem.
		\item We show that the above algorithm converges in finite time and provide its complexity analysis.
	\end{itemize}

	\subsection{Organization}
	We elaborate CNN-AE and its performance in Section \ref{section_nn} and explain how one can apply it for designing networks in Section \ref{network_design}.
	Although this article primarily targets cellular network and considers coverage and rate manifolds as  performance metrics, we argue in Section \ref{open_question} that extension of our framework to non -cellular scenarios and other performance metrics may also be feasible. We also discuss other extensions/applications of this work that might be interesting to both wireless communication and machine learning communities. Finally, we conclude our article in Section \ref{conclusion}.

	\section{NN Based Prediction}
	\label{section_nn}
	
	In this section, we first describe the datasets that are used in this work, along with the preprocessing steps that are applied on them before being fed to our proposed NN model. 
	Next, we explain our proposed architecture and its training and testing process. Finally, we discuss its performance in comparison to the SG-based models.  
	
	\subsection{Datasets}

	BS location data of India, Brazil, Germany and the USA are collected from the database \href{https://www.opencellid.org/downloads.php}{$\mathrm{www.opencellid.org}$}. We divide each country into square grids such that the size of the smallest square is $L\times L$\footnote{The spherical nature of the Earth must be taken into account while drawing such grids. If the Earth is presumed to be a perfect sphere with radius $R=6371$ km, then at a location with coordinates $(\theta, \phi)$, a small latitude change of $\Delta \theta$ would correspond to a geodesic length of $R\Delta \theta$ while a small longitude change of $\Delta \phi$ would correspond to a geodesic length of $R\cos(\theta)\Delta \phi$. Thus, at $(\theta, \phi)$, a square of size $L\times L$ would correspond to a rectangle with lengths $\Delta \theta = L/R$ and $\Delta\phi = L/R\cos(\theta)$ in the spherical coordinate space as long as $L \ll R$.}. We term each smallest square as a Region-of-Interest (RoI). We choose $L=10$ km for India and Brazil for their relatively lower BS densities and $L=5$ km for Germany and the USA for their relatively higher BS densities. Next, we count the number of BSs in each RoI and discard those RoIs that contain $20$ or less BSs. It enforces an interference-limited environment in each RoI. We also discard RoIs with $400$ or more BSs as those represent less than $1\%$ of the RoI population in each country. The remaining RoIs are used for training and testing of the NN. The (unfiltered) BS locations for India and Brazil can be visualized in Fig. \ref{fig:joint_vis}.
	
	\subsection{Performance Metrics and Ground Truth Generation}
	\label{sec_performance_simulation}
	
	In this paper, our primary target is to estimate the coverage and rate manifolds. However, the same framework can be used for other performance metrics as well. Below we describe how these quantities are calculated in a simulator to be utilised in the supervised learning.
	
	Let, $\mathcal{R}$ be an arbitrary RoI of size $L\times L$ and $\{\mathbf{r}_1,\cdots, \mathbf{r}_n\}$ be the locations of $n$ BSs contained in it. Let, $\mathbf{r}_0$ be a potential user location in $\mathcal{R}$. Any user located at $\mathbf{r}_0$ will receive signals from its nearest BS i.e., if $j_0$ is the index of the BS which the user is connected to, then
	\begin{align*}
		j_0 = \underset{j\in\{1,\cdots, n\}}{\arg\min} |\mathbf{r}_0-\mathbf{r}_j|
	\end{align*}
	
	It is presumed that each BS transmits at  power $P$  and reuses the same spectrum. Therefore the received signal strength at $\mathbf{r}_0$ can be expressed as: $Ph_{j_0}|\mathbf{r}_0-\mathbf{r}_{j_0}|^{-\alpha}$ where $h_{j_0}$ is a random variable indicating small scale fading and $\alpha$ is the path-loss coefficient. Similarly, the aggregate interference power at $\mathbf{r}_0$ generated by other BSs in $\mathcal{R}$ can be expressed as follows. 
	\begin{align}
		I_{\mathcal{R}} = \sum_{j\in\{1,\cdots,n\}\setminus j_0} Ph_{j}|\mathbf{r}_0-\mathbf{r}_j|^{-\alpha}
		\label{eq::interference_in_RoI}
	\end{align}
	where $\{h_j\}_{j\in\{1,\cdots, n\}}$ is a collection of independent and identically distributed random variables denoting small scale fading.
	
	\begin{figure*}[h!]
		\centering
		\includegraphics[width=.9\linewidth]{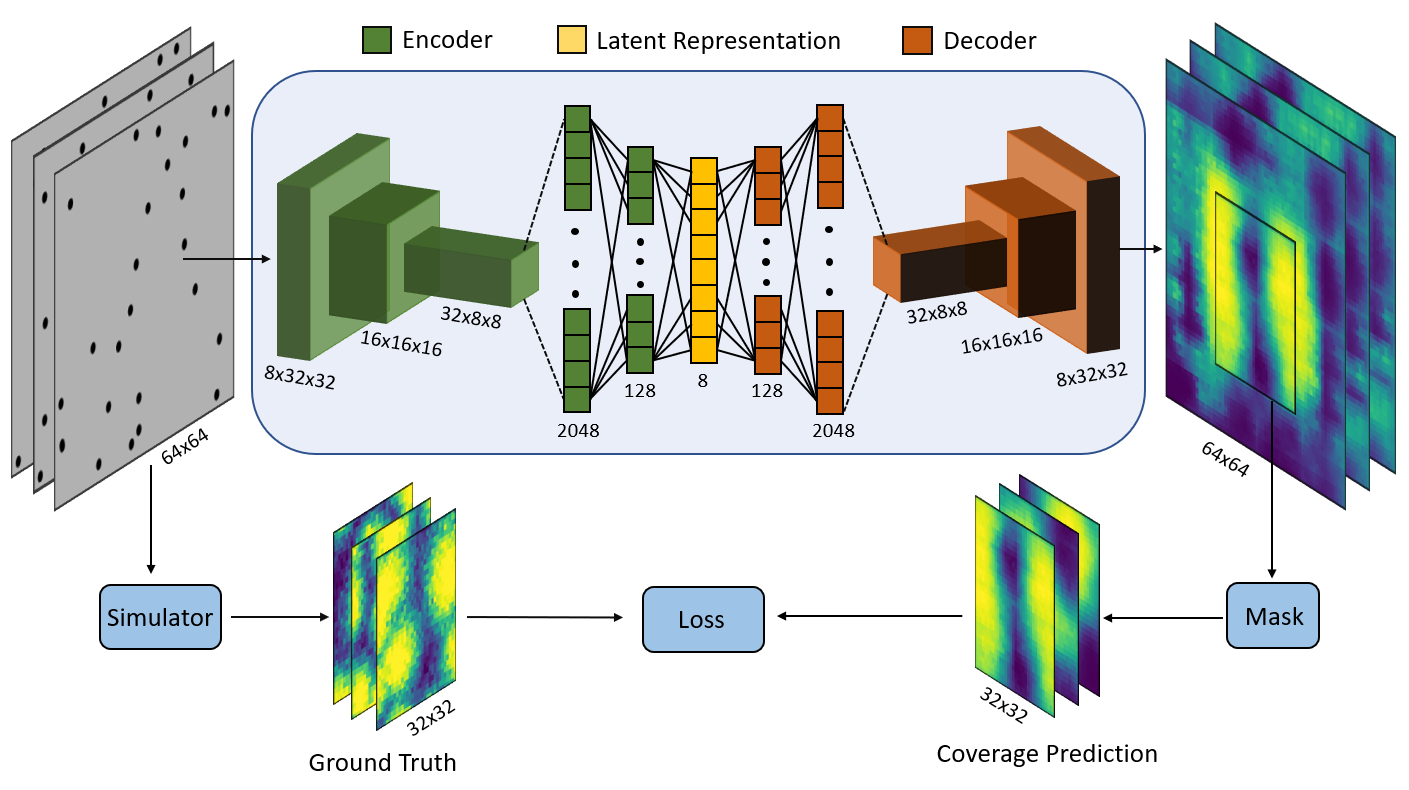}
		\caption{A schematic diagram of the CNN-AE architecture. Each convolution layer in the encoder performs convolution operation with filters of kernel $(3,3)$, stride $2$ and padding $1$. Also, the number of input and output channels of these layers are $(1, 8, 16)$ and $(8, 16, 32)$ respectively. The hyperparameters of the deconvolution filters in the decoder match exactly with their counterparts in the encoder. The hyperparameters of FF networks are depicted in the figure along with the size of the latent representations. The sizes of the output of each convolution layer and that of the input of each deconvolution layer are also shown in the figure. BS locations of the RoIs are fed to the NN as images of size  $64\times 64$. In response, it yields coverage manifolds of the same size. A $32\times32$ subset (corresponding to the RoE) of the generated manifold is then compared to its simulated counterpart. The obtained loss is used to update the NN-parameters using stochastic gradient-descent algorithm.}
		\label{fig:cnn}
	\end{figure*}
	
	We would like to clarify that the BSs that are located outside of $\mathcal{R}$ also contribute to the interference at $\mathbf{r}_0$. However, if $\mathbf{r}_0$ is sufficiently distant from the boundaries of $\mathcal{R}$, then the total interference generated from outside of $\mathcal{R}$ will be significantly less than that generated from inside of $\mathcal{R}$. In this case, $I_{\mathcal{R}}$ can be approximated to be the total interference at $\mathbf{r}_0$ and thus the (approximate) expression of SINR at $\mathbf{r}_0$ can be written as:
	\begin{align}
		\mathrm{SINR}(\mathbf{r}_{0})=\dfrac{h_{j_0}|\mathbf{r}_0-\mathbf{r}_{j_0}|^{-\alpha}}{\sum_{j\in\{1,\cdots,n\}\setminus j_0} h_j |\mathbf{r}_0-\mathbf{r}_j|^{-\alpha} + {\sigma^2}/{P}}
		\label{eq::sinr}	
	\end{align}
	where $\sigma^2$ indicates the power of additive white Gaussian noise (AWGN). The subset of an RoI where this approximation can be applied with reasonable accuracy is termed to be the Region of Evaluation (RoE). In this article, we assume that the RoE of an RoI of size $L\times L$ is a concentric square of size $L/2\times L/2$. The concept of RoE is important because the target of our NN is to approximate the mapping from BS locations to coverage and rate manifolds. To predict those manifolds within an RoI, we thus must restrict ourselves to a subset of RoI where those metrics are completely determined by the BS locations of the concerned RoI. The RoE subset exactly satisfies this criterion. We design the simulator such that, for a given RoI, it generates the coverage/rate manifolds only over its RoE. Other segments of the RoI are discarded as the interference from outside of the concerned RoI in those areas may be significant and therefore can no longer be ignored. The simulated manifolds associated with RoEs serve as the ground truth for training our NN.
	
	Using $(\ref{eq::sinr})$, we can obtain the coverage probability at $\mathbf{r}_0$ as,
	\begin{align}
		\begin{split}
			\mathrm{coverage}(\mathbf{r}_0, \gamma_{\mathrm{th}}) &= \mathbb{P}(\mathrm{SINR}(\mathbf{r}_0)>\gamma_{\mathrm{th}})
			\\
			&= \mathbb{E}\left[\mathbbm{1}\left(\mathrm{SINR}(\mathbf{r}_0)>\gamma_{\mathrm{th}}\right)\right]
		\end{split}	
		\label{eq::coverage}
	\end{align} 
	where $\gamma_{\mathrm{th}}$ is a pre-defined threshold and $\mathbbm{1}(.)$ is the indicator function. The ergodic rate at $\mathbf{r}_0$ can be computed as,
	\begin{align}
		\mathrm{rate}(\mathbf{r}_0) = \mathbb{E}\left[\log_2(1+\mathrm{SINR}(\mathbf{r}_0))\right]
		\label{eq::rate}
	\end{align}
	
	The expectations in $(\ref{eq::coverage}), (\ref{eq::rate})$ are taken over the distributions of fading gains. In simulator, these expectations are calculated via Monte Carlo simulations.

	Notice that the variable $\mathbf{r}_0$ is continuous. For the purpose of computing the coverage and rate manifolds, we must discretise the spatial dimensions i.e., the possible values that $\mathbf{r}_0$ can take. We discretise each RoI and its RoE into $64\times 64$ and $32\times 32$ square-grids respectively. Thus, the BS locations within an RoI and the coverage/rate manifolds over its RoE can be presented as a $64\times 64$ binary image and a $32\times 32$ color-map respectively. The discretisation levels can be increased to gain finer details, however, that comes at the cost of larger computation time.

	\subsection{Proposed NN Architecture}
	\label{section::nn_arcchitecture}

	In this work, we use a Convolutional Neural Network based Auto-Encoder (CNN-AE) to estimate the coverage and the rate manifolds of a network from its BS location data. We shall first describe the CNN-AE architecture for predicting the coverage manifold and later demonstrate how the same NN can be used for predicting the rate manifold with slight modifications in the training data. 
	
	Fig. \ref{fig:cnn} demonstrates the CNN-AE architecture for coverage prediction. Structurally, CNN-AE is segregated into two parts- encoder and decoder. At a high level, the job of an encoder is to generate a low-dimensional latent representation of the input RoI. On the other hand, the goal of a decoder is to generate the coverage manifold of the same RoI from its encoder-generated latent representation.
	
	The BS locations of an arbitrary RoI are fed into the encoder as  a $64\times 64$ binary image. It then passes through three convolutional layers, each containing a sublayer of convolution filters followed by a rectified linear unit ($\mathrm{ReLU}$) activation function. The output of the third convolutional layer is then straightened and fed into a fully connected feed-forward (FF) network with a hidden layer. The output of the FF network acts as the latent representation of the given RoI. All hyperparameters of CNN-AE are provided in Fig. \ref{fig:cnn}. \hspace{2cm}
	
	The structure of decoder is the exact opposite of that of the encoder. The latent representation is first fed into a FF network with one hidden layer. Its output is then arranged into a square grid and fed to a cascade of three deconvolutional layers, each containing deconvolution filters followed by a $\mathrm{ReLU}$ function. The last deconvolution layer also contains a $\mathrm{sigmoid}$ function to ensure that every $64\times 64$  points of its generated output lies between $0$ and $1$. This output is finally passed through a mask to yield a $32\times 32$ manifold corresponding to the RoE of the input.
	
	The NN shown in Fig. \ref{fig:cnn} can also be utilised to estimate the rate manifold. However, as each point of the output manifold lies in $[0,1]$, we must scale the rate appropriately. 
	
	\begin{figure*}[h!]
		\centering
		\begin{subfigure}{0.48\textwidth}
			\centering
			\includegraphics[width=\linewidth]{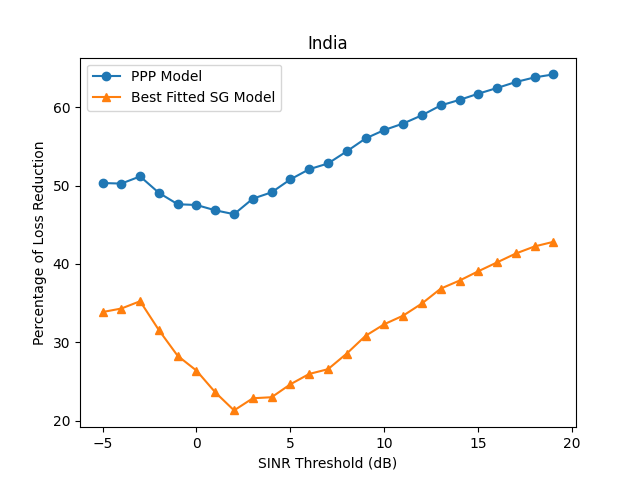}
			\label{subfig:err_red_india}
		\end{subfigure}
		\begin{subfigure}{0.48\textwidth}
			\centering
			\includegraphics[width=\linewidth]{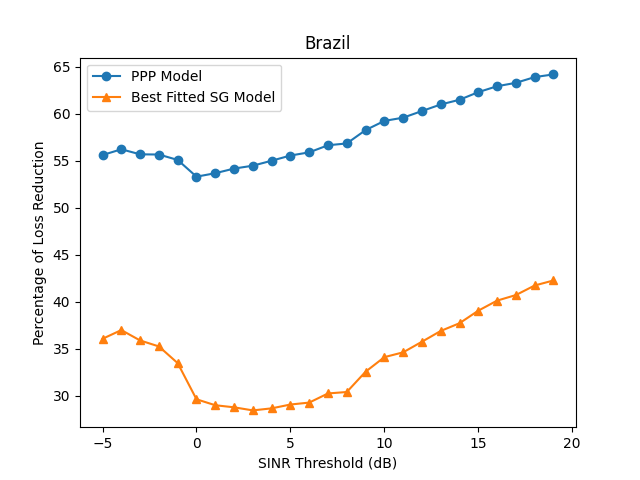}
			\label{subfig:err_red_brazil}
		\end{subfigure}
		\begin{subfigure}{0.48\textwidth}
			\centering
			\includegraphics[width=\linewidth]{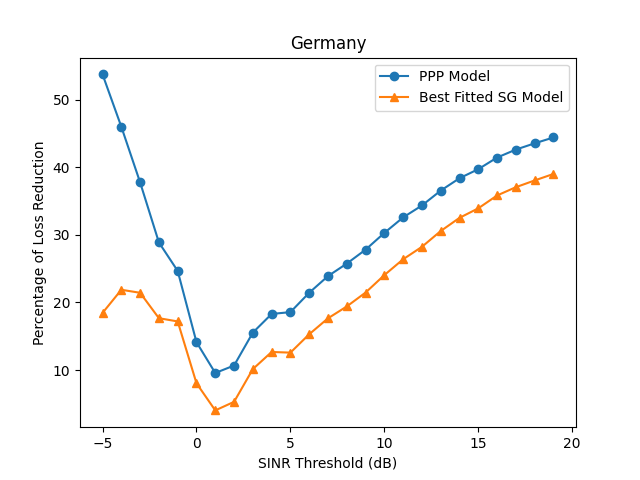}
			\label{subfig:err_red_germany}
		\end{subfigure}
		\begin{subfigure}{0.48\textwidth}
			\centering
			\includegraphics[width=\linewidth]{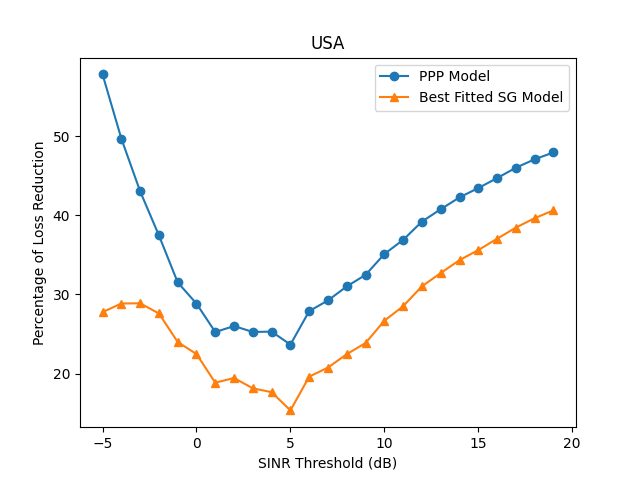}
			\label{subfig:err_red_usa}
		\end{subfigure}
		\caption{The relative error performance of CNN-AE in comparison to the PPP-based model and the best fitted SG-based model in estimating the coverage manifold in an RoE. The path-loss coefficient, $\alpha$, is taken as $4$ and the fading distributions are taken to be exponential with mean $1$ (Rayleigh fading). The results are averaged over $5$ random seed values.}
		\label{fig:err_red_results}
	\end{figure*}
	
	\subsection{Training and Testing Procedures}
	Out of all eligible RoIs within a country, $70\%$ are randomly chosen for training CNN-AE while the rest are used for testing its performance. For each RoI, its CNN-AE generated $32\times 32$ manifold is compared to its simulated counterpart and the loss is obtained via $L_1$ function which is a common loss-function in the image processing literature \cite{liu2020gradnet}. Particularly if $\{X_{i,j}\}_{i,j=1}^{32}$ denotes the simulated rate/coverage manifold corresponding to the RoE of an RoI, and $\{Y_{i,j}\}_{i,j=1}^{32}$ indicates its corresponding NN-based output (after masking), then loss is defined as below.
		\begin{align*}
			\mathrm{loss}\triangleq \sum_{i=1}^{32}\sum_{j=1}^{32}|X_{i,j}-Y_{i,j}|
		\end{align*}

	The computed loss is utilised to update the NN-parameters via a stochastic gradient descent (SGD) algorithm. 
	The testing procedure is similar to the training except for the fact that the NN-parameters are no longer updated.
	
	\subsection{Error Performance}
	
	In this section, we compare the average error/loss generated by our trained NN during the testing phase with that generated by the SG-based models. As clarified earlier, SG-based models cannot provide location-specific performance values. To obtain the error for these models, we first generate a $32\times 32$ manifold for each RoI such that each points in that manifold holds the same average performance value as predicted by the SG-model and then compare it with the simulated ground-truth manifold via $L_1$ function.
	
	We compare the performance of our model to two SG-based models. The first model is Poisson point process (PPP) which is one of the simplest and most widely adapted models in the literature. If the BSs are deployed in an infinite area following a PPP with density $\lambda$, then the expression of average coverage probability and rate can be obtained following Theorem 1 and 3 in  \cite{andrews2011tractable}. If the BS locations of an RoI are approximated as a realization of a PPP, then to compute the average performances, we must first estimate the BS density $\lambda$ using the estimator $(\ref{eq::lambda_estimator})$ and then plug it into their respective expressions. \hspace{1.5cm}
	\begin{align}
		\hat{\lambda} = \dfrac{\# \mathrm{~BSs ~in ~RoI}}{\mathrm{Area ~of ~RoI}}
		\label{eq::lambda_estimator}
	\end{align}
	
	\begin{figure*}[t]
		\centering
		\begin{subfigure}{0.48\textwidth}
			\centering
			\includegraphics[width=\linewidth]{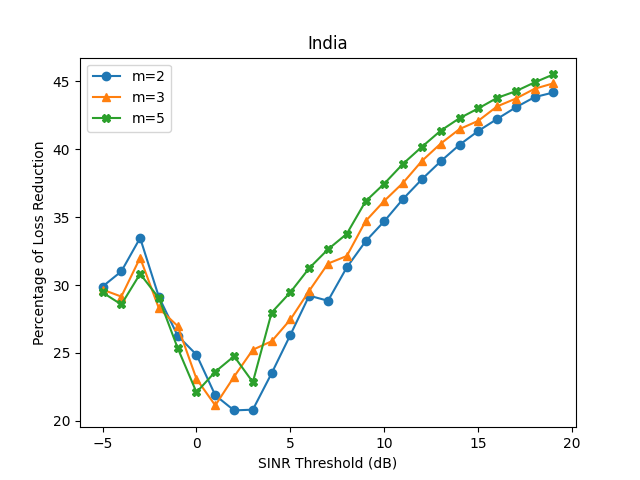}
			\label{subfig:nakagami_india}
		\end{subfigure}
		\begin{subfigure}{0.48\textwidth}
			\centering
			\includegraphics[width=\linewidth]{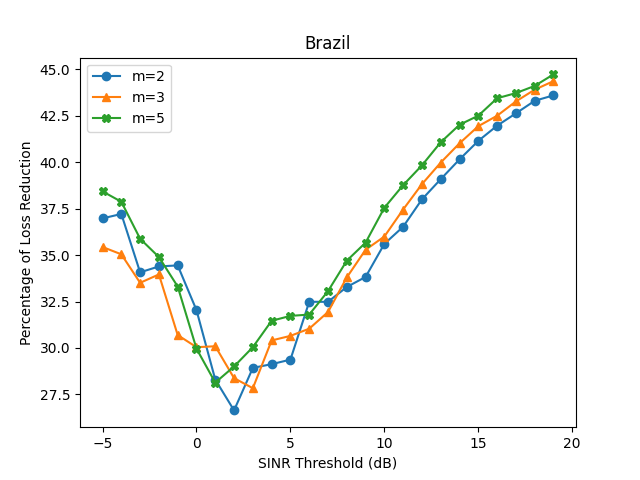}
			\label{subfig:nakagami_brazil}
		\end{subfigure}
		\begin{subfigure}{0.48\textwidth}
			\centering
			\includegraphics[width=\linewidth]{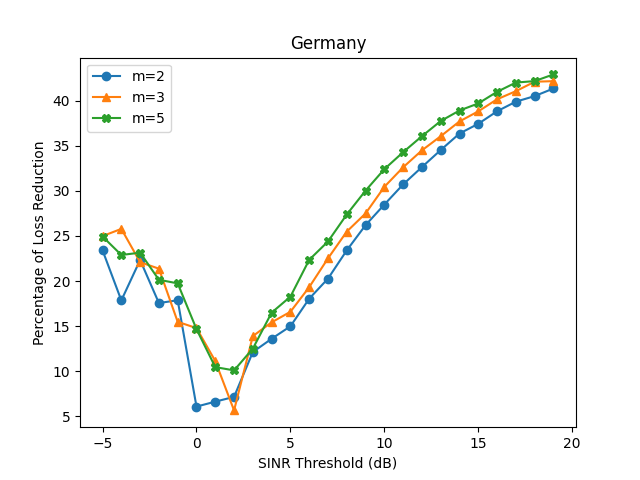}
			\label{subfig:nakagami_germany}
		\end{subfigure}
		\begin{subfigure}{0.48\textwidth}
			\centering
			\includegraphics[width=\linewidth]{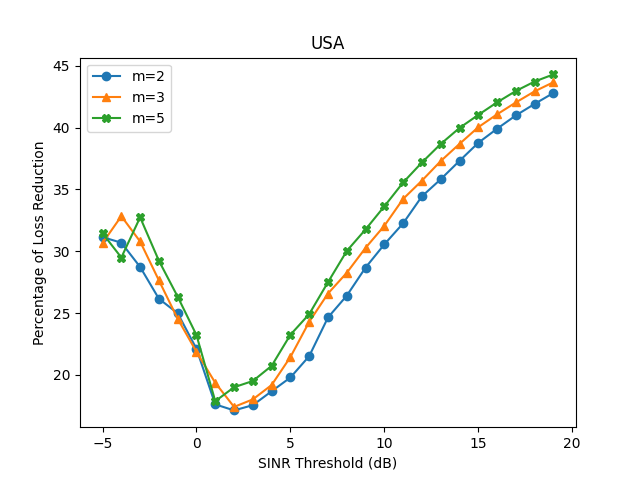}
			\label{subfig:nakagami_usa}
		\end{subfigure}
		\caption{The relative error performance of CNN-AE in comparison to the best fitted SG-based model in estimating the coverage manifold in an RoE. The channel gains are Gamma distributed (Nakagami fading) with shape parameter $m=2,3,5$ and scale parameter $1$. The pathloss coefficient $\alpha$ is taken to be $4$. The results are averaged over $5$ random seed values.}
		\label{fig:nakagami_results}
	\end{figure*}
	
	We denote the second SG-based model as the best-fitted SG model. In this case, we presume that, in each RoI, the average performance value predicted by the model is exactly the same as the empirical spatial average of the simulated performance manifold. In other words, we assume that the best-fitted model can predict the ground-truth average performance value with zero error. Clearly, the loss generated by this model is a lower bound to the loss performance of any other SG-based models.

	Fig. \ref{fig:err_red_results} describes the relative performance of our NN-model in predicting the coverage manifold in comparison to the two SG-based models stated above over a wide range of the SINR threshold, $\gamma_{\mathrm{th}}$.
	In particular, for a given  $\gamma_{\mathrm{th}}$, we mathematically define the relative performance/loss reduction of our NN-based model in comparison to the PPP-based model as follows. \hspace{0.5cm}
		\begin{align*}
			\mathrm{loss~reduction}(\mathrm{PPP}) = \dfrac{\mathrm{loss}(\mathrm{PPP})-\mathrm{loss}(\mathrm{NN})}{\mathrm{loss}(\mathrm{PPP})}\times 100\%
		\end{align*}
		where $\mathrm{loss}(\mathrm{PPP})$, $\mathrm{loss}(\mathrm{NN})$  indicate the losses corresponding to the PPP and the NN-based outputs respectively. The relative performance of the NN-based model in comparison to the best-fitted SG-based model is defined similarly.
	
	We observe that, in each country, the CNN-AE based model improves the error performance by a margin of as large as $40\%$ in comparison to the best-fitted model for some values of $\gamma_{\mathrm{th}}$. Moreover, over the range of all $\gamma_{\mathrm{th}}$ values considered in our result, the minimum values of error performance improvement, in comparison to the best-fitted model, are approximately $21\%$, $28\%$, $4\%$ and $15\%$ for India, Brazil, Germany and the USA respectively\footnote{The python code for generating these results are openly available at: \href{https://github.itap.purdue.edu/Clan-labs/CoverageRate_via_CNN_AE}{https://github.itap.purdue.edu/Clan-labs/CoverageRate\_via\_CNN\_AE}}.
	
	The improvements in comparison to the PPP-model are even more dramatic. In this scenario, over all $\gamma_{\mathrm{th}}$ values considered in our experiment, the minimum improvement values for India, Brazil, Germany and the USA are $46\%$, $53\%$, $10\%$, and $23\%$, respectively while the peak improvement values are $65\%$, $65\%$, $53\%$, and $57\%$, respectively.

	Interestingly, the performance improvement values are high at both ends of the $\gamma_{\mathrm{th}}$ spectrum and low in the middle. It can be explained as follows. At the ends, the coverage values at all the potential user locations are close to either $1$ (low end) or $0$ (high end) and hence present little uncertainty to the prediction algorithm. In the middle, however, coverage values are almost equidistant from both of its bounds and thereby present high uncertainty. Thus, predictive CNN-AE performs relatively better at the ends than in the middle. \hspace{4.0cm}
	
	In generating the above results, the BS to user channel gains are presumed to be Rayleigh faded. In Fig. \ref{fig:nakagami_results}, we consider the channel fading distribution to be Nakagami$-m$ that subsumes the Rayleigh distribution as its special case \cite{ibrahim2021exact}. Interestingly, we observe that, the relative error performance of the CNN-AE architecture in comparison to the best-fitted SG-based model in predicting the coverage manifold does not change significantly with change in $m$. In other words, the improvements are robust across a wide range of  fading environments. 
	
	Finally, in Table \ref{table:rate_loss_reduction}, we enlist the relative performance of the CNN-AE in predicting the rate manifold in comparison to the best-fitted SG-based model. For various fading parameters $m$, we see that the error improvement varies between $18\%-22\%$, $19\%-25\%$, $15\%-15\%$, $16\%-17\%$ for India, Brazil, Germany and the USA respectively. This is consistent with our previous	results (Fig. \ref{fig:nakagami_results}) that the estimation error improvement of CNN-AE is robust against changes in the fading environment. 
	
	\section{Application in Network Design}
	\label{network_design}
	
	In this section, we describe how a trained CNN-AE can be used to optimally place new BSs in an RoI while obeying some spatially heterogeneous coverage requirements within its RoE. We would like to clarify that, although our discussion shall be confined to the coverage-driven network design, it can also be extended to rate-driven designs as well. \hspace{3cm}
	
	\begin{table}
		\centering
		\begin{tabular}{|c|c|c|c|c|}
			\hline
			$m$ &India & Brazil & Germany & USA\\
			\hline
			1 & $22.44\%$ & $25.29\%$ & $15.44\%$ & $16.98\%$\\
			\hline 
			2 &   $20.15\%$        &   $23.58\%$        &    $15.85\%$       &  $17.83\%$        \\
			\hline
			3 &   $18.25\%$        &   $22.85\%$        &    $15.47\%$       &     $17.23\%$     \\
			\hline
			5 &   $18.27\%$        &   $19.06\%$        &   $15.45\%$        &      $18.24\%$    \\
			\hline
		\end{tabular}
		\caption{The reduction of error of CNN-AE in comparison to the best fitted SG-based model in estimating the rate manifolds in RoEs. The channel gains are taken to be Gamma distributed (Nakagami fading) with shape parameters $m\in\{1, 2, 3, 5\}$ and scale parameter $1$ while the pathloss coefficient, $\alpha$, is presumed to be $4$. The results are averaged over $5$ seed values. }
		\label{table:rate_loss_reduction}
	\end{table}
	
	To better understand our objective, consider the situation of a network designer who is faced with the problem of deploying new BSs in an RoI which can be segregated into a high demand area where the coverage must be provided at least $90\%$ of the time (i.e., the coverage probability must exceed $0.9$) and a low demand area where it is sufficient to provide coverage $80\%$ of the time.  The designer is satisfied with a deployment solution if these constraints are met at at least $95\%$ of the user locations. Algorithm \ref{alg:alg_1} describes a method to solve such questions. More specifically, its goal is to provide solution to the problem of optimally deploying new BSs in an RoI, alongside the existing ones (brownfield deployment), such that the coverage manifold generated within the RoE by the new network satisfies a pre-defined location-specific constraint at at least a certain fraction of the potential user locations.

		The algorithm initiates with one new BS (for loop initiation at line \ref{for_loop_4}). Notice that the RoI is discretised into a $N\times N$ grid\footnote{In section \ref{sec_performance_simulation}, we have taken $N=64$.}. Our first task is to check whether the BS in consideration can be deployed at one of the $N^2$  locations such that its associated coverage manifold exceeds certain  location-specific thresholds at at least a certain fraction of user locations. The first part of the job is in handled by the subroutine $\mathrm{CyclicOpt}$. Specifically, the $\mathrm{for}$ loop at line \ref{for_25} loops the potential location of this BS all over the $N\times N$ grid. For each location, its associated coverage manifold is obtained by feeding the network topology (location of BS) to a trained NN (line \ref{line_28}). We then calculate the fraction
		of potential user locations where the computed coverage values exceed given location specific thresholds (line \ref{line_29},). Subroutine $\mathrm{CyclicOpt}$ returns the maximum value of the fraction (denoted by $\mathrm{CycleMaxFrac}$) calculated over all possible $N^2$ locations. In line \ref{line_14}, we compare the returned fraction value with a given threshold $\mathrm{FracTh}$. If it does not exceed $\mathrm{FracTh}$, we increment the number of deployable BSs by one (progress in the $\mathrm{for}$ loop at line $\ref{for_loop_4}$) and move to the next stage. Otherwise, the algorithm terminates and returns a  favourable deployment solution (line \ref{line_18}).

	\begin{algorithm}[ht]
		\caption{$\mathrm{Brownfield ~Network~ Design~in~an~RoI}$}
		\begin{algorithmic}[1]
			\State \textbf{Inputs:} 
			
			$\mathrm{OldBSLoc}$ \Comment{Existing BS Locations in RoI} 
			
			$\mathrm{CNN-AE}$ \Comment{Trained NN for a given SINR Threshold}
			
			$N$ \Comment{Spatial discretization levels of RoI}
			
			$X_{\max}, Y_{\max}$ \Comment{Dimensions of RoI}
			
			$\mathrm{MaxBS}$ \Comment{Maximum BSs that can be deployed}
			
			$\mathrm{CovTh}$  \Comment{Location Specific Coverage Threshold}
			
			$\mathrm{FracTh}$ \Comment{Minimum fraction of locations in RoE that must have coverage above their specified threshold} 
			
			\vspace{0.2cm}
			
			\State $\mathrm{MaxFrac}\gets 0$
			\State $(\Delta x , \Delta y) \gets (X_{\max}/N, Y_{\max}/N)$ \Comment{Spatial resolutions}
			\vspace{0.2cm}
			\For{$\mathrm{NumBS}\in\{1, \cdots, \mathrm{MaxBS}\}$}
			\label{for_loop_4}
			\State $\mathrm{NewBSLoc}\gets$ $\mathrm{NumBS}$ random locations in RoI.
			\vspace{0.1cm}
			\While{$\mathrm{True}$}
			\label{while-loop}
			\State $\mathrm{CycleMaxFrac}, \mathrm{CycleOptLoc}$ $\gets$ $\mathrm{CyclicOpt}( \hspace{0.05cm})$
			\vspace{0.1cm}
			\If{$\mathrm{CycleMaxFrac}>\mathrm{MaxFrac}$}
			\State $\mathrm{MaxFrac}\gets \mathrm{CycleMaxFrac}$
			\State $\mathrm{NewBSLoc}\gets \mathrm{CycleOptLoc}$
			\State $\mathrm{OptNewBSLoc}\gets \mathrm{CycleOptLoc}$
			\Else
			\State Break while loop
			\label{line_13}
			\EndIf
			
			\If{$\mathrm{MaxFrac}\geq\mathrm{FracTh}$}
			\label{line_14}
			\State Break while loop
			\State Break for loop
			\EndIf
			\EndWhile
			\EndFor
			\If{$\mathrm{MaxFrac}<\mathrm{FracTh}$}
			\label{line_17}
			\State $\mathrm{OptNewBSLoc}\gets\mathrm{NONE}$
			
			\EndIf
			\State \textbf{Output:} $\mathrm{OptNewBSLoc}$ 
			\label{line_18}
			\vspace{0.3cm}

			\Procedure{$\mathrm{CyclicOpt}$}{ }
			\State $\mathrm{CycleMaxFrac}\gets 0$
			\State $\mathrm{CycleOptLoc}\gets \mathrm{NewBSLoc}$
			\State $\mathrm{Grid}\gets\{0,\Delta x, \cdots, X_{\max}\}\times\{0,\Delta y, \cdots, Y_{\max}\}$
			
			\vspace{0.1cm}
			\For{$j\in\{0,\cdots, \mathrm{NumBS}-1\}$}
			\State $\mathrm{tempLoc}\gets \mathrm{CycleOptLoc}$
			\For{$(x, y)\in\mathrm{Grid}$}
			\label{for_25}
			\State $\mathrm{tempLoc}[j]\gets (x, y)$
			\State $\mathrm{Topology}\gets \mathrm{OldBSLoc}\cup\mathrm{tempLoc}$
			\State $\mathrm{CovManifold}\gets \mathrm{CNN-AE}(\mathrm{Topology})$
			\label{line_28}
			\State $\mathrm{Frac}\gets \mathrm{mean}(\mathrm{CovManifold}>\mathrm{CovTh})$
			\label{line_29}
			\vspace{0.1cm}
			\If{$\mathrm{Frac}>\mathrm{CycleMaxFrac}$}
			\State $\mathrm{CycleMaxFrac}\gets\mathrm{Frac}$
			\State $\mathrm{CycleOptLoc}\gets \mathrm{tempLoc}$ \vspace{0.2cm}
			\EndIf
			\EndFor
			\EndFor
			\State	\Return $\mathrm{CycleMaxFrac}$, $\mathrm{CycleOptLoc}$
			\EndProcedure
		\end{algorithmic}
		
		\label{alg:alg_1}
	\end{algorithm}

	\begin{figure*}
		\centering
		\includegraphics[width=0.95\linewidth]{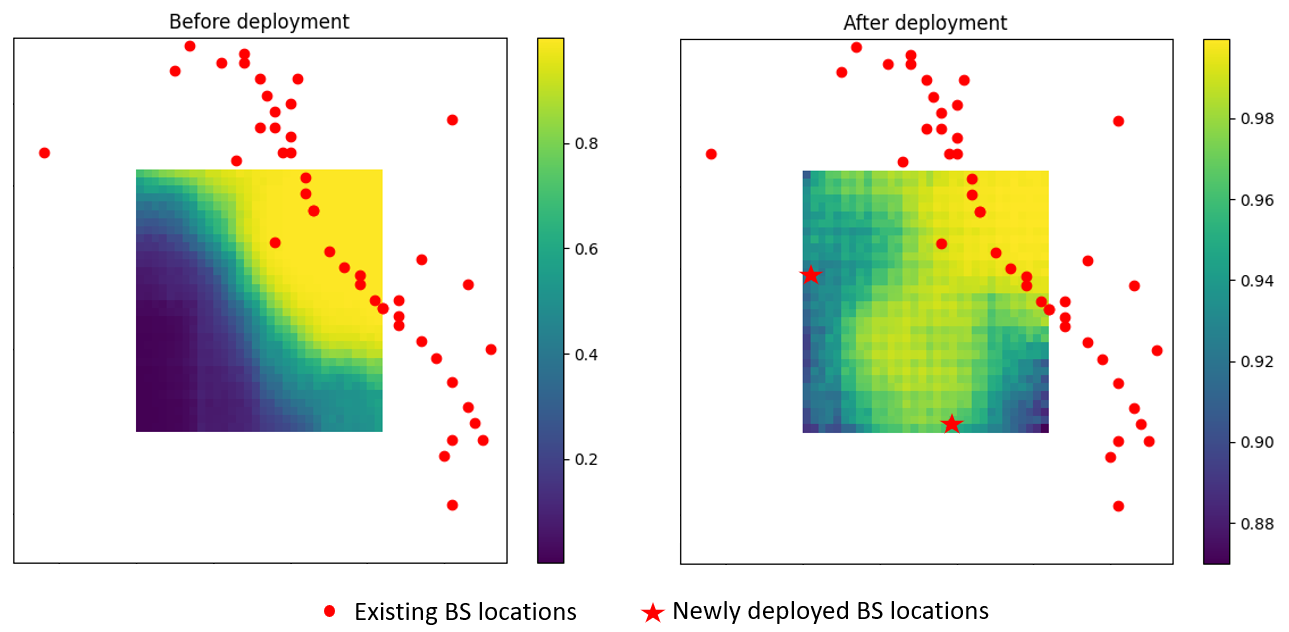}
		\caption{The coverage probability manifolds associated with a random RoI of India before and after deploying some new BSs. The BSs are deployed using Algorithm \ref{alg:alg_1} with the constraint that the coverage probability must exceed $0.9$ at at least $95\%$ of the user locations. The SINR threshold value $\gamma_{\mathrm{th}}$, is taken to be $0$ dB for computing the coverage probabilities. We would like to point out that the scale of colors in the two subfigures are different.}
		\label{fig:before_after_deployment}
	\end{figure*}
	In the next stage, the goal of the algorithm is to find suitable deployment locations of the two new BSs. However, due to the combinatorial nature of the problem, the number of possible scenarios is $\mathcal{O}(N^4)$. In general, for $k$ new BSs, the number of possibilities increases to $\mathcal{O}(N^{2k})$. Thus, unlike the first stage, the exhaustive search can no longer be used. To circumvent the exponential blowup, we now apply the alternate maximization process \cite{mondal2021economic} (while loop at line \ref{while-loop}). The main idea is described as follows. For given locations of new BSs, we first obtain the coverage manifold by feeding the entire topology to a trained CNN-AE (line \ref{line_28}) 
	and calculate the fraction of potential user locations where the generated manifold satisfies a pre-defined location dependent lower bound (line \ref{line_29}). Alternate or Cyclic maximization process maximizes this fraction by unilaterally varying the location of only one deployable BS at a time.
	The BS whose location is to be varied is chosen in a cyclic manner. 
	This cyclic process continues unless at the end of a cycle any one of the following two termination conditions gets validated. 
	The first termination flag triggers if the maximum value of the fraction obtained at the current cycle turns out to be less than that obtained in the previous cycle (line \ref{line_13}). 
	It indicates that the cyclic procedure can no longer improve the fraction value.
	In this case, the number of deployable new BSs is incremented by one and the cycle starts all over again (progress in for loop at line \ref{for_loop_4}).
	The other termination flag activates if the
	value of the fraction obtained at the current cycle exceeds a
	pre-defined threshold (line \ref{line_14}). 
	In this case, the whole program terminates and returns a viable solution.
	
	Note that Algorithm \ref{alg:alg_1} takes $\mathrm{MaxBS}$, the maximum number of deployable BSs, as an input. Therefore, if no viable solution is obtained even after deploying $\mathrm{MaxBS}$ number of new BSs, (line \ref{line_17}) then the program returns $\mathrm{NONE}$  and terminates.

	We would like to mention that the cyclic optimization process ($\mathrm{while}$ loop at line \ref{while-loop}) can potentially steer the Algorithm \ref{alg:alg_1} into an infinite loop. However, in the following theorem, we prove that such a scenario cannot arise and thus the algorithm is guaranteed to terminate within a finite amount of time.
	
	\begin{theorem}
		Algorithm \ref{alg:alg_1} is guaranteed to terminate within a finite number of steps.
	\end{theorem}
	\begin{proof}
		It is sufficient to establish that the while loop in line $\ref{while-loop}$ always terminates in finite steps. Note that, for a given level of discretization $N$, and $\mathrm{NumBS}$ number of deployable BSs, the number of new network scenarios can be at most $(N^{2})^{\mathrm{NumBS}}$, which is a finite number.  Consequently, the number of distinct outcomes of the subroutine $\mathrm{CyclicOpt}$ must at most be finite. Note that the while loop continues as long as $\mathrm{CycleMaxFrac}$, the first component of the outcome of the process $\mathrm{CyclicOpt}$ exceeds all of its previous occurrences in the while loop. Thus, if the loop never terminates, then the values of  $\mathrm{CycleMaxFrac}$ obtained in different iterations of the while loop must form a strictly increasing sequence. This is an impossibility since only finitely many possibilities  are allowed for the elements of that sequence.
	\end{proof}
	
	The computational complexity of the subroutine $\mathrm{CyclicOpt}$ is $\mathcal{O}(N^2)$ where $N$ denotes the number of spatial discretisation levels in RoI. If $k$ indicates the maximum number of iterations of the $\mathrm{while}$ loop in line $\ref{while-loop}$, then the computational complexity of Algorithm \ref{alg:alg_1} can be written as $\mathcal{O}(k\mathrm{MaxBS}N^2)$ where the term $\mathrm{MaxBS}$ defines the maximum number of deployable BSs. Although, in theory, $k$ can be large, our numerical experiments exhibit that in practice, $k$ is small. Effectively, the complexity of Algorithm \ref{alg:alg_1}, therefore, can be written as $\mathcal{O}(\mathrm{MaxBS}N^2)$.
	
	We would like to point out that there are many heuristic and meta-heuristic algorithms in the literature (see \cite{liu2019efficient} and the references therein) to deploy BSs in a network. Unlike our algorithm, however, they do not use NN to generate the coverage/rate manifolds. Moreover, the objectives and constraints in those articles are different than ours. For example, \cite{liu2019efficient} places BSs to maximize the coverage area of the network.
	
	We apply Algorithm \ref{alg:alg_1} to a randomly chosen RoI of India to deploy some new BSs to elevate the coverage probability to $0.9$ for at least $95\%$ of the user locations in the RoE. Fig. \ref{fig:before_after_deployment} exhibits before and after deployment scenarios.  Similar experiment can also  be performed with RoIs of other countries.

	\section{Open Questions}
	\label{open_question}
	
	Our article lays out the foundation of NN-based translation of network topology to its performance manifold. Interestingly, there are many avenues to extend this work. Below we discuss some of these possibilities.
	
	\subsubsection{Other Network Scenarios} Our article primarily analyses outdoor cellular networks. The channel models are thus chosen accordingly. However, the same concepts can also be extended to indoor networks (e.g., Visible Light Communication \cite{costanzo2021adaptive}), device-to-device communication \cite{khuntia2021bidirectional}, vehicular networks \cite{posner2021federated} etc. Each of these network scenarios present its own challenge in designing appropriate NN because the channel models, the network topology and the potential receiver locations could be drastically different from one another.
	
	\subsubsection{Variants of Cellular Networks} There are many variants of the cellular network itself where our work can be extended. For example, in this paper, the network is taken to be homogeneous, i.e., all BSs are presumed to be identical. Unfortunately, in reality, BSs can differ in height, transmission directionality, signal power etc. Incorporating this network heterogeneity into a unified NN-framework is an important future endeavor. One possible way to tackle the heterogeneity of transmission power could be to represent the BS locations by numbers in $[0, 1]$ that are proportional to the strength of the transmitted signal. It is in stark  contrast with our current approach where each location in the RoI is represented by a binary random variable indicating the presence/absence of a BS at that specific location.
	
	\subsubsection{Performance Metrics} In our work, we take coverage and rate as the network performance metrics. However, in the fifth generation (5G) and beyond networks, many other metrics are deemed important. For example, in ultra reliable low latency communication (URLLC), latency is one of the most important performance criteria \cite{ali2021urllc}. Moreover, in sensor networks where majority of the devices may be battery-driven, power efficiency is an important requirement \cite{ansere2019reliable}. NN-aided prediction of these performance metrics could be an interesting area to explore.
	
	\subsubsection{Architectural Modification} 
	
	We utilise a  NN-architecture that is commonly used for image-to-image mapping. However, many variations of this NN are possible. For example we could increase the depth of encoder-decoder, increase the size of the hidden layers etc. Understanding how the hyperparameters can effect the error performance is essential to design better NNs to improve the accuracy.
	
	\subsubsection{Training Data} We use the output of a simulator as the ground-truth for training the NN. Instead, if country-wide field measurements are used, then it can potentially aid the training process in two ways.  First, it can dramatically reduce the time for computation because a significant portion of it is used for simulating the ground truth. Secondly, the channel model used in the simulation cannot capture the finer details of the reality.

	In summary, our work points towards a myriad of opportunities for both the machine learning community and the wireless network community to contribute.
	
	\section{Conclusion}
	\label{conclusion}
	
	In this article, we design a  CNN-AE to predict the coverage and rate manifolds of a network from its topology. We train our  model by feeding BS location data of India, Brazil, Germany, and the USA. In comparison to the stochastic geometry based baseline model, CNN-AE reduces coverage and rate prediction error by a margin of as large as $40\%$ and $25\%$ respectively. As an application, we show how trained CNN-AE can be used for brownfield network design. We also discuss
	how our approach can be extended to other application areas. \hspace{1.0cm}

	
	\ifCLASSOPTIONcaptionsoff
	\newpage
	\fi
	
	\bibliographystyle{IEEEtran}
	\bibliography{Bib}

\begin{thebibliography}{10}
\providecommand{\url}[1]{#1}
\csname url@samestyle\endcsname
\providecommand{\newblock}{\relax}
\providecommand{\bibinfo}[2]{#2}
\providecommand{\BIBentrySTDinterwordspacing}{\spaceskip=0pt\relax}
\providecommand{\BIBentryALTinterwordstretchfactor}{4}
\providecommand{\BIBentryALTinterwordspacing}{\spaceskip=\fontdimen2\font plus
\BIBentryALTinterwordstretchfactor\fontdimen3\font minus
  \fontdimen4\font\relax}
\providecommand{\BIBforeignlanguage}[2]{{%
\expandafter\ifx\csname l@#1\endcsname\relax
\typeout{** WARNING: IEEEtran.bst: No hyphenation pattern has been}%
\typeout{** loaded for the language `#1'. Using the pattern for}%
\typeout{** the default language instead.}%
\else
\language=\csname l@#1\endcsname
\fi
#2}}
\providecommand{\BIBdecl}{\relax}
\BIBdecl

\bibitem{lu2021stochastic}
X.~Lu, M.~Salehi, M.~Haenggi, E.~Hossain, and H.~Jiang, ``Stochastic geometry
  analysis of spatial-temporal performance in wireless networks: A tutorial,''
  \emph{IEEE Communications Surveys \& Tutorials}, 2021.

\bibitem{wijethilaka2021survey}
S.~Wijethilaka and M.~Liyanage, ``Survey on network slicing for internet of
  things realization in {5G} networks,'' \emph{IEEE Communications Surveys \&
  Tutorials}, vol.~23, no.~2, pp. 957--994, 2021.

\bibitem{ozcan2020robust}
Y.~{\"O}zcan, J.~Oueis, C.~Rosenberg, R.~Stanica, and F.~Valois, ``Robust
  planning and operation of multi-cell homogeneous and heterogeneous
  networks,'' \emph{IEEE Transactions on Network and Service Management},
  vol.~17, no.~3, pp. 1805--1821, 2020.

\bibitem{sabbah2018emulation}
A.~Sabbah, A.~Jarwan, I.~Al-Shiab, M.~Ibnkahla, and M.~Wang, ``Emulation of
  large-scale {LTE} networks in {NS}-3 and core: A distributed approach,'' in
  \emph{MILCOM 2018-2018 IEEE Military Communications Conference
  (MILCOM)}.\hskip 1em plus 0.5em minus 0.4em\relax IEEE, 2018, pp. 1--6.

\bibitem{andrews2011tractable}
J.~G. Andrews, F.~Baccelli, and R.~K. Ganti, ``A tractable approach to coverage
  and rate in cellular networks,'' \emph{IEEE Transactions on communications},
  vol.~59, no.~11, pp. 3122--3134, 2011.

\bibitem{chen2018modeling}
C.~Chen, R.~C. Elliott, W.~A. Krzymie{\'n}, and J.~Melzer, ``Modeling of
  cellular networks using stationary and nonstationary point processes,''
  \emph{IEEE Access}, vol.~6, pp. 47\,144--47\,162, 2018.

\bibitem{li2018stochastic}
R.~Li, Z.~Zhao, Y.~Zhong, C.~Qi, and H.~Zhang, ``The stochastic geometry
  analyses of cellular networks with {$\alpha$}-stable self-similarity,''
  \emph{IEEE Transactions on Communications}, vol.~67, no.~3, pp. 2487--2503,
  2018.

\bibitem{deng2014ginibre}
N.~Deng, W.~Zhou, and M.~Haenggi, ``The {G}inibre point process as a model for
  wireless networks with repulsion,'' \emph{IEEE Transactions on Wireless
  Communications}, vol.~14, no.~1, pp. 107--121, 2014.

\bibitem{afshang2018poisson}
M.~Afshang and H.~S. Dhillon, ``Poisson cluster process based analysis of
  {H}et{N}ets with correlated user and base station locations,'' \emph{IEEE
  Transactions on Wireless Communications}, vol.~17, no.~4, pp. 2417--2431,
  2018.

\bibitem{otter2020survey}
D.~W. Otter, J.~R. Medina, and J.~K. Kalita, ``A survey of the usages of deep
  learning for natural language processing,'' \emph{IEEE Transactions on Neural
  Networks and Learning Systems}, vol.~32, no.~2, pp. 604--624, 2020.

\bibitem{xie2020adversarial}
C.~Xie, M.~Tan, B.~Gong, J.~Wang, A.~L. Yuille, and Q.~V. Le, ``Adversarial
  examples improve image recognition,'' in \emph{Proceedings of the IEEE/CVF
  Conference on Computer Vision and Pattern Recognition}, 2020, pp. 819--828.

\bibitem{yuan2020machine}
J.~Yuan, H.~Q. Ngo, and M.~Matthaiou, ``Machine learning-based channel
  prediction in massive {MIMO} with channel aging,'' \emph{IEEE Transactions on
  Wireless Communications}, vol.~19, no.~5, pp. 2960--2973, 2020.

\bibitem{zhu2021uav}
B.~Zhu, E.~Bedeer, H.~H. Nguyen, R.~Barton, and J.~Henry, ``{UAV} trajectory
  planning in wireless sensor networks for energy consumption minimization by
  deep reinforcement learning,'' \emph{IEEE Transactions on Vehicular
  Technology}, vol.~70, no.~9, pp. 9540--9554, 2021.

\bibitem{xu2021experience}
J.~Xu and B.~Ai, ``Experience-driven power allocation using multi-agent deep
  reinforcement learning for millimeter-wave high-speed railway systems,''
  \emph{IEEE Transactions on Intelligent Transportation Systems}, 2021.

\bibitem{liu2020gradnet}
Y.~Liu, S.~Anwar, L.~Zheng, and Q.~Tian, ``Gradnet image denoising,'' in
  \emph{Proceedings of the IEEE/CVF Conference on Computer Vision and Pattern
  Recognition Workshops}, 2020, pp. 508--509.

\bibitem{ibrahim2021exact}
H.~Ibrahim, H.~Tabassum, and U.~T. Nguyen, ``Exact coverage analysis of
  intelligent reflecting surfaces with {N}akagami-m channels,'' \emph{IEEE
  Transactions on Vehicular Technology}, vol.~70, no.~1, pp. 1072--1076, 2021.

\bibitem{mondal2021economic}
W.~U. Mondal, A.~A. Sardar, and G.~Das, ``Economic analysis of cognitive
  underlay networks: A {N}ash bargaining based approach,'' \emph{IEEE
  Transactions on Vehicular Technology}, vol.~70, no.~2, pp. 2024--2029, 2021.

\bibitem{liu2019efficient}
Y.~Liu, W.~Huangfu, H.~Zhang, H.~Wang, W.~An, and K.~Long, ``An efficient
  geometry-induced genetic algorithm for base station placement in cellular
  networks,'' \emph{IEEE Access}, vol.~7, pp. 108\,604--108\,616, 2019.

\bibitem{costanzo2021adaptive}
A.~Costanzo, V.~Loscri, and M.~Biagi, ``Adaptive modulation control for visible
  light communication systems,'' \emph{Journal of Lightwave Technology},
  vol.~39, no.~9, pp. 2780--2789, 2021.

\bibitem{khuntia2021bidirectional}
P.~Khuntia, R.~Hazra, and P.~Goswami, ``A bidirectional relay-assisted underlay
  device-to-device communication in cellular networks: An {IoT} application for
  fintech,'' \emph{IEEE Internet of Things Journal}, 2021.

\bibitem{posner2021federated}
J.~Posner, L.~Tseng, M.~Aloqaily, and Y.~Jararweh, ``Federated learning in
  vehicular networks: opportunities and solutions,'' \emph{IEEE Network},
  vol.~35, no.~2, pp. 152--159, 2021.

\bibitem{ali2021urllc}
R.~Ali, Y.~B. Zikria, A.~K. Bashir, S.~Garg, and H.~S. Kim, ``{URLLC} for {5G}
  and beyond: Requirements, enabling incumbent technologies and network
  intelligence,'' \emph{IEEE Access}, vol.~9, pp. 67\,064--67\,095, 2021.

\bibitem{ansere2019reliable}
J.~A. Ansere, G.~Han, H.~Wang, C.~Choi, and C.~Wu, ``A reliable energy
  efficient dynamic spectrum sensing for cognitive radio {IoT} networks,''
  \emph{IEEE Internet of Things Journal}, vol.~6, no.~4, pp. 6748--6759, 2019.

\end{thebibliography}
 
\begin{IEEEbiography}
[{\includegraphics[width=1in, height=1.25in, clip, keepaspectratio]{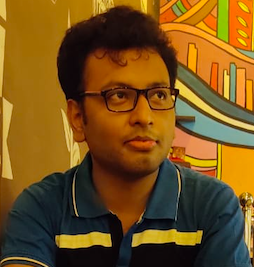}}]
    {Washim Uddin Mondal} is a post-doctoral researcher at Purdue University, USA. He received a Ph. D. from IIT Kharagpur, India under the Prime Minister's Research Fellowship (PMRF) scheme in 2021. He obtained his B. Tech-M.Tech dual degree in ECE from the same institute in 2016. His research interests are in game theory and machine learning. He has recently won the best paper award at the NeurIPS workshop of Cooperative AI, 2021.
\end{IEEEbiography}

\begin{IEEEbiography}
[{\includegraphics[width=1in, height=1.25in, clip, keepaspectratio]{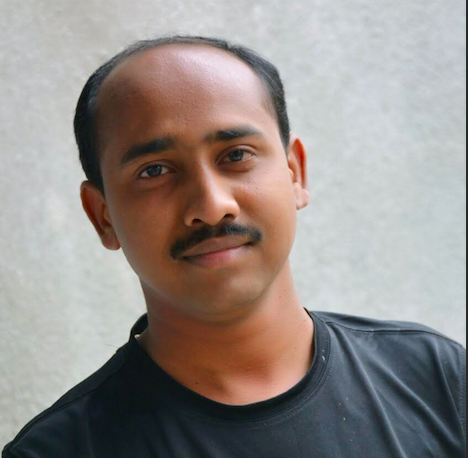}}]
{Praful D. Mankar (Member, IEEE)} received the bachelor's degree in electrical and communication engineering from Amravati University, India, in 2006. He has done a master's in telecommunications systems engineering and a Ph.D. in wireless networks from the Indian Institute of Technology (IIT) Kharagpur, India, in 2009 and 2016, respectively. He also worked as a Post-doctoral Research Associate with Wireless@Virginia Tech Research Group at Virginia Tech, Blacksburg, VA, USA, from 2017 to 2019. Since 2020, he has been working as an Assistant Professor in the Signal Processing and Communication Research Center at the International Institute of Information Technology Hyderabad (IIIT-H), India. His research interests are wireless networks, stochastic geometry, reflecting intelligent surfaces, and the age of information.
\end{IEEEbiography}

\begin{IEEEbiography}
[{\includegraphics[width=1in, height=1.25in, clip, keepaspectratio]{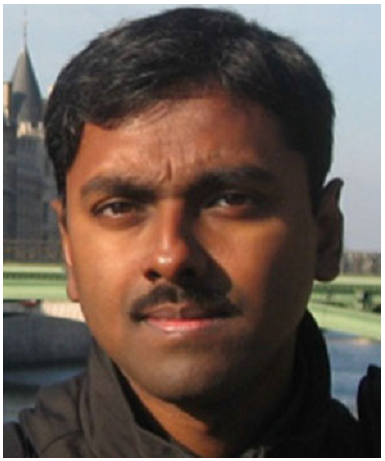}}]
{Goutam Das} 
received the Ph.D. degree from the University of Melbourne, Australia, in 2008. He was a Post-Doctoral Fellow with Ghent University, Belgium, from 2009 to 2011. He is currently as an Assistant Professor with the Indian Institute of Technology Kharagpur. His research interests include optical access networks, radio over fiber technology, optical packet-switched networks, media access protocol design, game theory, and economic analysis of access networks. He has served as a member of the organizing committee of IEEE ANTS since 2011.
\end{IEEEbiography}

\begin{IEEEbiography}[{\includegraphics[width=1in,height=1.25in,clip,keepaspectratio]{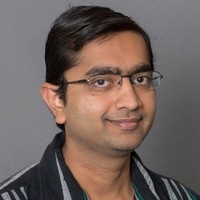}}]{Vaneet Aggarwal (S'08 - M'11 - SM'15)}
received the B.Tech. degree in 2005 from the Indian Institute of Technology, Kanpur, India, and the M.A. and Ph.D. degrees in 2007 and 2010, respectively from Princeton University, Princeton, NJ, USA, all in Electrical Engineering.

He is currently a Full Professor at Purdue University, West Lafayette, IN, where he has
been since Jan 2015. He was a Senior Member of
Technical Staff Research at AT\&T Labs-Research,
NJ (2010-2014), Adjunct Assistant Professor at
Columbia University, NY (2013-2014),  VAJRA Adjunct Professor at IISc
Bangalore (2018-2019), and is currently a Visiting Professor at KAUST, Saudi Arabia. His current research interests are in machine learning and its applications.

Dr. Aggarwal was the recipient of Princeton University’s Porter Ogden Jacobus Honorific Fellowship in 2009, 2017 Jack Neubauer Memorial Award
recognizing the Best Systems Paper published in the IEEE Transactions on
Vehicular Technology,  2018 Infocom Workshop Best Paper Award, and 2021 NeurIPS Cooperative AI Workshop Best Paper Award. He was on the Editorial Board of IEEE Transactions on Green Communications and Networking (2017-2020) and the IEEE Transactions on Communications (2016-2021), and is currently on the Editorial Board of the IEEE/ACM Transactions on Networking (2019-current), and the founding co-Editor-in-Chief of the ACM Journal on Transportation Systems (2022-current). 
\end{IEEEbiography}

\begin{IEEEbiography}
[{\includegraphics[width=1in, height=1.25in, clip, keepaspectratio]{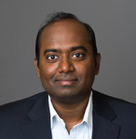}}]
{Satish V. Ukkusuri} is a Reilly Professor in the Lyles School of Civil Engineering at Purdue University and Director of the Urban Mobility Networks and Intelligence (UMNI) Lab. His research is in the area of interdisciplinary transportation networks with current interests in data-driven mobility solutions, disaster management, the resilience of interdependent networks, connected and autonomous traffic systems, shared mobility platforms, and smart logistics. He is a University Faculty Scholar (2017-present), ASCE Fellow, Fulbright Fellow, a selectee of the National Academy of Engineering (NAE) JAFOE conference (2016), a selectee of the National Academy of Science (NAS) Arab American Frontiers of Science, Engineering and Medicine in 2017 and a CUTC/ARTBA Faculty Award (2011) among other awards. He has published more than 350 peer-reviewed papers and is the Editor in Chief of the Journal of Big Data Analytics in Transportation and co-Editor in Chief of ACM Journal of Autonomous Transportation.  
\end{IEEEbiography}

\end{document}